\documentclass{article}

\usepackage[all]{xy}
\usepackage{natbib}
\usepackage[margin=1in]{geometry}
\usepackage{fullpage}
\usepackage{graphicx}
\usepackage{hyperref}
\usepackage{tikz}
\usepackage{enumitem}
\usepackage{xcolor}
\usepackage[utf8]{inputenc}
\usepackage{booktabs}
\hypersetup{
   breaklinks,%
   colorlinks=true,%
   linkcolor=black,%
   urlcolor=black,%
   citecolor=[rgb]{0,0,0.45}
}

\usepackage{amsmath}
\usepackage{amssymb}
\usepackage{amsthm}
\newtheorem{theorem}{Theorem}[section]

\newtheorem{lemma}[theorem]{Lemma}

\newtheorem{claim}[theorem]{Claim}

\newtheorem{definition}[theorem]{Definition}

\usepackage{algorithmic}
\usepackage{algorithm}

\begin{document}

\title{New Approximation Algorithms for Fair $k$-median Problem}

\author{Di Wu\thanks{School of Computer Science and Engineering, Central South University, Changsha, China. Email: \texttt{csuwudi@csu.edu.cn}}
\and Qilong Feng\thanks{School of Computer Science and Engineering, Central South University, Changsha, China. Email: \texttt{csufeng@mail.csu.edu.cn}}
\and Jianxin Wang\thanks{School of Computer Science and Engineering, Central South University, Changsha, China. Email: \texttt{jxwang@mail.csu.edu.cn}}}
\date{}
\maketitle
\begin{abstract}
The fair $k$-median problem is one of the important clustering problems. The current best approximation ratio is 4.675 for this problem with 1-fair violation, which was proposed by Bercea et al. [APPROX-RANDOM'2019]. As far as we know, there is no available approximation algorithm for the problem without any fair violation.
In this paper, we consider the fair $k$-median problem in bounded doubling metrics and general metrics. We provide the first QPTAS for fair $k$-median problem in doubling metrics. Based on the split-tree decomposition of doubling metrics, we present a dynamic programming process to find the candidate centers, and apply min-cost max-flow method to deal with the assignment of clients. Especially, for overcoming the difficulties caused by the fair constraints, we construct an auxiliary graph and use minimum weighted perfect matching to get part of the cost. For the fair $k$-median problem in general metrics, we present an approximation algorithm with ratio $O(\log k)$, which is based on the embedding of given space into tree metrics, and the dynamic programming method. Our two approximation algorithms for the fair $k$-median problem are the first results for the corresponding problems without any fair violation, respectively.
\end{abstract}

\section{Introduction}\label{sec1}

As one of the most commonly studied problems, clustering aims to give a "good" partition of a set of points of a metric space into different groups so that points with similar attributes should be in the same group.
The clustering algorithm may be biased if the data set involves some sensitive attributes (gender, race, or age). To solve this deviation caused by some sensitive attributes of the data set, the \emph{fair} clustering problem was proposed.  Over the past few years, lots of attention have been paid on the fair clustering problem (\cite{abbasi2021fair, ahmadian2020fair, backurs2019scalable, chen2019proportionally,
Chierichetti0LV17, ghadiri2021socially, BeraCFN19, esmaeili2020probabilistic}).

\cite{Chierichetti0LV17} first proposed the definition of \emph{redistricting fairness} for clustering, which refers to the balance of disparate impact and fair representation of each protected class in every cluster. \cite{Bercea0KKRS019} first proposed a generalized and tunable notion of redistricting fairness where asserts that protected classes should maintain an upper bound and a lower bound proportion of the cluster size in clustering. This kind of fairness model is called $(\alpha,\beta)$-proportional fair model. The input for fair $k$-median problem is a set of points $\mathcal {C}$, a set of candidate centers $\mathcal{F}$ in a metric space together with an integer $l$, and two vectors $\boldsymbol \alpha$ and $\boldsymbol \beta$, the fair $k$-median problem asks for a set $F\subseteq \mathcal{F}$ of $k$ points, called $\emph centers$, together with an assignment $\mu:\mathcal{C}\rightarrow F$, where each client in the data set $\mathcal {C}$ is entitled a color $color(i)$ and $\mathcal {C}$ is divided into $l$ disjoint groups $\mathcal {C}_1,\mathcal {C}_2,\ldots,\mathcal {C}_l$, the goal is to minimize the $k$-median objective function satisfying that the number of each type clients within $color(i)$ in each cluster conforms to a certain proportion between $\alpha_i$ and $\beta_i$.

\cite{Bercea0KKRS019} first proposed a bicriteria approximation scheme with $4.675$-approximation and $1$-fairness violation in polynomial time for the fair $k$-median problem in metric space, which can be extended to a $3/3.488/62.856$-approximation for the fair $k$-center/facility location/$k$-means problem within a $1$-fairness violation, respectively.
For the special case of $\alpha_i=\alpha, \beta_i=0$ ($\alpha$ is a fixed number and $i\in [l]$), \cite{ahmadian2019clustering} gave a $3$-approximation with $2$-fairness violation algorithm for fair $k$-center problem in metric space.
As a generalization for the above $(\alpha,\beta)$-proportional fairness model, \cite{BeraCFN19} proposed a more generalized notion of fairness such that each client can belong to multiple types. They gave an approximation algorithm with $\rho+2$-approximation and $4\Delta+3$-fairness violation in polynomial time for the fair $k$-median, fair $k$-center and fair $k$-means problems, where $\rho$ is the approximation ratio of the known $k$-clustering algorithm, and $\Delta$ denotes the maximum number of types that each client belongs to. \cite{harb2020kfc} gave a $3$-approximation and $4\Delta+3$-fairness violation for the fair $k$-center problem in polynomial time.  As far as we know, for the fair $k$-median problem, no approximation algorithm without fairness violation is available for any non-trivial metrics.

\section{Results and Techniques}\label{sec2}

We propose the first qusi-polynomial time approximation scheme(QPTAS) for the fair $k$-median problem in doubling metrics, the first $O(\log k)$-approximation algorithm for the fair $k$-median problem in general metrics, and these two results have no fairness violation.

\begin{theorem}
Given an instance of size $n$ of the fair $k$-median problem in metrics of fixed doubling dimension $d$, there exists a randomized $(1+\epsilon)$-approximation algorithm with running time $\tilde{O}(n^{{(\frac{\log n}{\epsilon})}^{O(d)}l})$.
\end{theorem}

\begin{theorem}
Given an instance of size $n$ of the fair $k$-median problem in general metric space, there exists a randomized $O(\log k)$-approximation algorithm with running time $n^{O(l)}$.
\end{theorem}


Our method solving the fair $k$-median problem in doubling metrics is based on the dynamic programming process on split-tree decomposition. Firstly, the split-tree decomposition is to partition the metrics into a number of subsets randomly, called \emph{blocks}, and keeps dividing these blocks recursively until each block only contains one point. For each block, we construct a subset of points, called portals, as the bridges to get the distance of points in split-tree. Our goal is to prove that there exists a near-optimal solution such that different blocks "interact" only through portals, which is one of the major challenges to solve the fair $k$-median problem in doubling metrics. The dynamic programming algorithm proceeds on the split-tree from the leaves to the root to fill in the DP table. For a table entry of a specific block, we save the information that how many clients of each type located outside (respectively inside) is assigned to facilities located inside (respectively outside) this block through each portal, and the number of facilities opened in this block. For a subproblem of a specific block, the value of the table entry of this block is calculated by considering the following two parts: the total value of table entries of the children of this block which can be extracted from the DP table; the cost between the portals of this block and the ones in its children.
How to calculate the cost between the portals of a block and the ones in its children is another challenge to solve the fair $k$-median problem in doubling metrics. To solve this problem, we construct an auxiliary graph and use minimum weighted perfect matching to get the cost between the portals of a block and the ones in its children. The dynamic programming procedure returns a set of facilities that need to be opened, and the number of clients assigned to each open facility satisfying the fairness constraints. Finally, we apply the min-cost max-flow method to assign the clients to the facilities.

The algorithm solving the fair $k$-median problem in general metric space is based on the embedding of given space to tree metrics such that a Hierarchically Separated Tree (HST) can be obtained. Moreover, a dynamic programming process is presented on the Hierarchically Separated Tree to get the candidate opened facilities.

\section{Preliminaries}\label{sec3}

\begin{definition}[metric space]\label{kmwp}
  Given a space $\mathcal {X}=(X, dist)$ where $X$ is a set of points associated with a distance function $dist:X\times X\rightarrow \Re_0$, space $\mathcal {X}$ is called a metric space if for any point $x, y, z$ in $X$, the following properties are satisfied: (1) $dist(x,y)=0$ iff $x=y$; (2) $dist(x,y)=dist(y,x)$; (3) $dist(x,y)+dist(y,z) \geq dist(x,z)$.
\end{definition}

\begin{definition}[doubling dimension]
Given a metric space $(X,dist)$ and a non-negative real number $r$, for any point $x\in X$, let $B(x, r)=\{y\in X\mid dist(x, y)\leq r\}$ denote the ball around $x$ with radius $r$. The doubling dimension of the metric space $(X,dist)$ is the smallest integer $d$ such that any ball $B(x, 2r)$ can be covered by at most $2^d$ ball $B(x, r)$. A metric space with doubling dimension is called a doubling metric space.
\end{definition}

Given a metric space $\mathcal {X}=(X,dist)$, let $\Delta$ be the ratio between the largest and the smallest distance in $\mathcal {X}$, which is called the \emph{aspect ratio} of $\mathcal {X}$. For any subset $Y\subseteq X$, the \emph{aspect ratio} of $Y$ is defined as $\frac{\max_{y,y'\in Y} d(y,y')}{\min_{y,y'\in Y} d(y,y')}$. For a positive integer $l$, let $[l]=\{1,2,...,l\}$.

\begin{definition}[the fair $k$-median problem]
Given a metric space $(X, dist)$ with a set of clients $\mathcal{C}$ and a set of facilities  $\mathcal{F}$, a non-negative integer $k$ and two vectors $\boldsymbol \alpha,\boldsymbol \beta\in \mathcal{R}^l$ $(0\leq\alpha_i\leq\beta_i\leq1)$, assume that $\mathcal{C}$ is partitioned into $l$ disjoint groups, i.e., $\mathcal{C}=\{\mathcal{C}_1,\mathcal{C}_2,\ldots,\mathcal{C}_l\}$ and $\mathcal{C}_i\cap \mathcal{C}_j = \emptyset$ $(i\neq j)$. For any $i\in [l]$, if $x\in \mathcal{C}_i$, $x$ is called an $i$-th type client. The fair $k$-median problem is to find a subset $F\subseteq \mathcal {F}$ of size at most $k$ and an assignment function $\mu : \mathcal {C}\rightarrow F$ such that:
(1) for any $f\in F$, $i\in[l]$, $\alpha_i\leq \frac{\mid\{c\in \mathcal{C}_i\mid \mu(c)=f\}\mid}{\mid\{c\in \mathcal{C}\mid \mu(c)=f\}\mid}\leq \beta_i$; (2) $\sum_{c\in \mathcal{C}} dist(c,\mu (c))$ is minimized.
\end{definition}


For an instance $(X, dist, \mathcal{C}, \mathcal{F}, k, l,\boldsymbol \alpha, \boldsymbol \beta)$ of the fair $k$-median problem, given a solution $(F,\mu)$ of instance $(X, dist, \mathcal{C}, \mathcal{F}, k,l, \boldsymbol \alpha, \boldsymbol \beta)$, the facilities in $F$ are called candidate facilities. For a facility $f\in F$ and  a client $c \in C$, if $\mu(c)=f$, then it is called that $f$ \emph{serves} client $c$. Let $X=\mathcal{C}\cup\mathcal{F}$, and assume that the size of $X$ is $n$.

\section{A QPTAS for Fair $k$-median Problem in Doubling Metrics}

For the fair $k$-median problem in doubling metrics, we first construct a \emph {split-tree}, and design a dynamic programming procedure to get the solution.

\subsection{Split-tree Decomposition} \label{Sec 3.1}

Given an instance $(X, dist, \mathcal{C}, \mathcal{F}, k,l, \boldsymbol \alpha, \boldsymbol \beta)$ of the fair $k$-median problem, following the assumptions in \cite{behsaz2019approximation}, the aspect-ratio $\Delta$ of the input metrics is at most $O(n^4/\epsilon)$ (where $n$ is the size of $\mathcal{C} \cup \mathcal{F}$, and $\epsilon$ is a constant). In order to satisfy the aspect-ratio assumption, we first deal with the points in $X$ by the following preprocessing steps and we use $OPT$ to denote an optimal solution of fair $k$-median problem.
Find an $O(1)$-approximation feasible solution $L$ of the fair $k$-median problem (\cite{Bercea0KKRS019}), and let $cost(L)$ be the cost of the feasible solution obtained. Then, if there exists a pair of points $x,y\in X$ with distance less than $\epsilon cost(L)/n^4 $, remove $x$, copy a point $x'$ of $x$, and add $x'$ at $y$ (note that if there are two clients at $y$, the two clients may not necessarily be assigned to the same facility in the solution). In the new instance obtained, a point $x$ is within distance at most $n \cdot \epsilon cost(L)/n^4$ from its original location. Thus, the cost of any solution for this new instance compared with the original instance is increased by at most $ n \cdot \epsilon cost(L)/n^3 \leq \epsilon cost(OPT)$.

For any set $Y$ of points, a subset $S\subseteq Y$ is called a $\rho $-$ covering$ of $Y$ if for any point $ y\in Y $, there is a point $s\in S$ such that $dist(y,s)\leq \rho$. A subset $S\subseteq Y$ is called a $\rho $-$ packing$ of $Y$ if for any $s,s' \in S$,\ $dist(s,s')\geq \rho$. A subset $S\subseteq Y$ is a $\rho $-$ net$ in $Y$ if it is both a $\rho $-$ covering$ of $Y$ and a $\rho$-$ packing$ of $Y$. The size of a net in metrics with fixed doubling dimension $d$ is bounded by the following lemma.

\begin{lemma}
\label{L4}
Let $(X,dist)$ be a metric space with doubling dimension $d$ and aspect-ratio $\Delta$, and let $ S\subseteq X $ be a $\rho $-$ net$. Then $\mid S\mid \leq (\frac {\Delta}{\rho})^ {O(d)}$.
\end{lemma}

A \emph{decomposition} of the metric $(X,dist)$ is a partition of $X$ into subsets, which are called \emph{blocks}. A \emph{hierarchical decomposition} is a sequence of $\ell+1$ decompositions $H_0,H_1,...,H_\ell$ such that every block of $H_{i+1}$ is the union of blocks in $H_i$ and $H_{\ell}=\{X\}$ and $H_0=\{\{x\}\mid x\in X\}$. The blocks of $H\_i$ are called the \emph{level} $i$ blocks.
A \emph{split-tree} of the metric space is a complete hierarchical decomposition: the root node is the level $\ell$ block $H_{\ell}=\{X\}$ and the leaves are the singletons such that $H_0=\{\{x\}\mid x\in X\}$.
Given a metric space $(X,dist)$, let $\Delta=2^\ell$ be the diameter of the metrics. We first construct a sequence of sets $Y_{\ell-1}\subseteq \ldots \subseteq Y_2\subseteq Y_1\subseteq Y_0=X$ such that $Y_i$ is a $2^{i-2}$-net of $Y_{i-1}$.
Starting from $H_{\ell}=\{X\}$, we now discuss the relation between the blocks in ($i+1$)-th level and $i$-th level, where $0\leq i\leq \ell-1$. Let $r_i=2^i \varrho$ where $\frac{1}{2}\leq \varrho<1$, and let $\pi$ be a random ordering of the points in $X$. For each point $x\in X$, let $\pi(x)$ be the order of $x$ in $\pi$. For the set $Y_i$, let $\{y_1, \ldots, y_h\}$ be the set of points in $Y_i$ with order from left to right in $\pi$.
For the point $y_1$, a new block $B^{y_1}$ can be obtained at level $i$ such that $B^{y_1}=\{x\in X\mid dist(x,y_1)\leq r_i\}$.
For any point $y_j$ $(1< j\leq h)$, a new block $B^{y_j}$ can be obtained at level $i$ by the following way: let $Z=\{x\in X\mid dist(x,y_j)\leq r_i\}$.
For each point $x$ in $Z$, if $x$ is not contained in any blocks $\{B^{y_1}, \ldots, B^{y_{j-1}}\}$, then $x$ is contained in $B^{y_j}$.

\begin{lemma}[\cite{talwar2004bypassing}]
\label{splitpro}
Given a split-tree of metric space $(X,dist)$ with a sequence of decompositions $H_0,H_1,...,H_\ell$, the split-tree decomposition has the following properties:
\begin{enumerate}
\item The total number of levels $\ell$ is $O(\log n)$ (since $\Delta=O(n^4/\epsilon)$).
\item Each block of level $i$ has diameter at most $2^{i+1}$, namely the maximum distance between any pair of points in a block of level $i$ is at most $2^{i+1}$.
\item Each level $i$ block is the union of at most $2^{O(d)}$ level $i-1$ blocks.
\item For any pair of points $u,v \in X$, the probability that they are in different sets corresponding to blocks at level $i$ of split-tree is at most $O(d)\cdot \frac{dist(u,v)}{2^i}$.
\end{enumerate}
\end{lemma}

Given an instance $(X, dist, \mathcal{C}, \mathcal{F}, k,l, \boldsymbol \alpha, \boldsymbol \beta)$ of the fair $k$-median problem, let $T$ be the split-tree obtained by using the methods in \cite{talwar2004bypassing}. For each block $B$ of level $i$ in $T$, we compute a $\rho2^{i+1}$-net $P$ of block $B$. Each point in $P$ is called a portal, and $P$ is also called a portal set. By Lemma \ref{L4}, it follows that the number of portals at a given block is $\rho^{-O(d)}$.
Moreover, the split-tree within portal set can be found in time $(1/\rho)^{O(d)}n\log \Delta$ (\cite{bartal2013linear,cohen2019near}). For a portal set $P$ of a block $B$ of level $i$ in $T$, there is an important property as follows.
\begin{lemma}[\cite{bartal2016traveling}]
\label{lportal}
For a portal set $P$ of a block $B$ of level $i$ in $T$, assume that the children of $B$ at level $i-1$ is $B_1,B_2,\ldots,B_u$ and each block $B_j$ has a portal set $P_j\subseteq B_j$. $P$ is a subset of the portal sets computed for the descendant blocks of $B$. That is, $P\subseteq P_1\cup P_2\cup \ldots\ P_u$.
\end{lemma}

Given a split-tree $T$ and two blocks $B_1$ and $B_2$ at level $i$ of $T$, and for two points $u\in B_1$, and $v\in B_2$,  the distance between $u$ and $v$ on the split-tree is defined as the length of the path which is constructed by the subpath from $u$ to a portal $p_i$ of $B_1$, the subpath from $p_i$ to a portal $p_j$ of $B_2$, and the subpath from $p_i$ to $v$. Given a block $B$ at level $i$ of split-tree $T$ and a portal set $P\subseteq B$, if there is a client $c_1$ outside of $B$ that is assigned to a facility $f_1$ inside of $B$ crossing at a portal $p_1\in P$, then we say that client $c_1$ enters $B$ through portal $p_1$. Similarly, if there is a client $c_2$ inside of $B$ that is assigned to a facility $f_2$ outside of $B$ crossing at a portal $p_2\in P$, then we say that client $c_2$ leaves $B$ through portal $p_2$. In the following, all the distances considered are the distances on the split-tree $T$.

\begin{lemma}
\label{5}
For any metric $(X,dist)$ with doubling dimension $d$ and any $\rho >0$, given a randomized split-tree $T$, a pair of blocks $B_1$ and $B_2$ of level $i$, a $\rho2^{i+1}$-net $P_1$ of block $B_1$, a $\rho2^{i+1}$-net $P_2$ of block $B_2$, and any two points $u\in B_1$ and $v\in B_2$, for the distance of $u$ and $v$ on the split-tree $T$ and the distance $dist(u, v)$ in the metric space, we have:
$\min_{p_i\in P_1,p_j\in P_2} \{dist(u,p_i)+dist(p_i,p_j)+dist(p_j,v)\} \leq dist(u,v)+O(\rho 2^{i+1})$.

\begin{proof}

By the triangle inequality, we have
\begin{equation*}
\begin{aligned}
dist(p_i,p_j)\leq dist(u,p_i)+dist(u,p_j)
             \leq dist(u,p_i)+dist(u,v)+dist(v,p_j).
\end{aligned}
\end{equation*}

\begin{equation*}
\begin{aligned}
dist(u,p_i)+dist(p_i,p_j)+dist(v,p_j)\leq 2dist(u,p_i)+dist(u,v)+2dist(v,p_j).
\end{aligned}
\end{equation*}
Then, we have
\begin{equation*}
\begin{aligned}
&\min_{p_i\in P_1,p_j\in P_2} \{dist(u,p_i)+dist(p_i,p_j)+dist(p_j,v)\}\\ &\leq\min_{p_i\in P_1,p_j\in P_2} \{2dist(u,p_i)+dist(u,v)+2dist(v,p_j)\}\\
             &=\min_{p_i\in P_1,p_j\in P_2} \{2dist(u,p_i)+2dist(v,p_j)\}\!+\!dist(u,v)\\
             &=\min_{p_i\in P_1,p_j\in P_2}2dist(u,p_i)\!+\!\min_{p_i\in P_1,p_j\in P_2}2dist(v,p_j) +dist(u,v).
\end{aligned}
\end{equation*}
Since $P_1$ is a $\rho2^{i+1}$-net of block $B_1$ and $P_2$ is a $\rho2^{i+1}$-net of block $B_2$, we have that $\min\limits_{p_i\in P_1,p_j\in P_2} 2 dist(u,p_i) \leq 2 dist(u,p_i')\leq 2\cdot\rho2^{i+1}$, where $p_i'$ is any point in $P_1$. Similarly, we have that $\min\limits_{p_i\in P_1,p_j\in P_2} 2 dist(u,p_i) \leq 2\cdot\rho2^{i+1}$. Thus, we have $$\min_{p_i\in P_1,p_j\in P_2} \{dist(u,p_i)+dist(p_i,p_j)+dist(p_j,v)\} \leq dist(u,v)+O(\rho 2^{i+1})$$
\end{proof}

\end{lemma}

\subsection{Algorithm for Fair $k$-median Problem in Doubling Metrics}\label{S 3.2}

In this section, we show how to design the dynamic programming process based on the split-tree.

\subsubsection{Dynamic Programming State}
Given a split-tree $T$, the dynamic programming algorithm proceeds on split-tree $T$ from the leaves to the root. For any block $B$ of split-tree $T$, we define the subtree rooted at block $B$ as $T_B$, and each subtree is a subproblem in our dynamic programming process, which includes the information that how many clients of each type entering or leaving the subtree through the portals of the subtree.

A table entry in the dynamic programming process is a tuple $DP[B,k_B,\boldsymbol Q_1^E,\boldsymbol Q_1^L,\boldsymbol Q_2^E, \boldsymbol Q_2^L,\ldots,\boldsymbol Q_m^E,\boldsymbol Q_m^L]$ where the parameters are defined as follows:
\begin{enumerate}
\item $m$ is the size of portal set in $B$ and $m=\rho^{-O(d)}$.
\item $B$ is the root node of the subtree $T_B$.
\item $k_B$($0\leq k_B\leq k$) is the number of open facilities in $B$.
\item $\boldsymbol Q_i^E$ ($i\in [m]$) is a vector with $l$ values and for $t\in[l]$, $q_t^{E(i)}$($0\leq q_t^{E(i)}\leq n$) in $\boldsymbol Q_i^E$ denotes the number of type $t$ clients that enter $B$ through the $i$-th portal.
\item $\boldsymbol Q_i^L$ ($i\in [m]$) is a vector with $l$ values and for $t\in[l]$, $q_t^{L(i)}$($-n\leq q_t^{L(i)}\leq 0$) in $\boldsymbol Q_i^L$ denotes the number of type $t$ clients that leave $B$ through the $i$-th portal.
\end{enumerate}

The cost of table entry $DP[B,k_B,\boldsymbol Q_1^E,\boldsymbol Q_1^L,\boldsymbol Q_2^E,\boldsymbol Q_2^L,\ldots,\boldsymbol Q_m^E,\boldsymbol Q_m^L]$ consists of the following three parts: (1) The cost of assigning clients inside of $B$ to facilities inside of $B$; (2) For the clients inside of $B$ assigned to facilities outside of $B$ through portals of $B$, the cost from the clients inside of $B$ to portals of $B$; (3) For the clients outside of $B$ assigned to facilities inside of $B$ through portals of $B$, the cost from portals of $B$ to the assigned facilities inside of $B$. 

The solutions at the root $R$ of $T$ are in the table entry $DP[R,k,\boldsymbol 0,\boldsymbol 0,\ldots,\boldsymbol 0]$, where $\boldsymbol 0$ is a vector with $l$ zero components. Among all these solutions in $DP[R,k,\boldsymbol 0,\boldsymbol 0,\ldots,\boldsymbol 0]$, the algorithm outputs the one with the minimum cost.

The base case of the dynamic programming is located at the leaves (which are singletons) of split-tree. For a leaf of split-tree, the corresponding block has only one portal and we only need to set $\boldsymbol Q_1^E$ and $\boldsymbol Q_1^L$, while the remaining $\boldsymbol Q_i^E$'s and $\boldsymbol Q_i^L$'s can be assigned $\boldsymbol 0$. Since there is only one facility or client in each leaf of split-tree, we consider the following three cases for the table entry of leaf node $B$:
\begin{enumerate}
\item If there is only one facility in the block and this facility is opened, then the table value of this block is $DP[B,1,\boldsymbol Q_1^E,\boldsymbol 0,\boldsymbol 0,\boldsymbol 0,\ldots,\boldsymbol 0,\boldsymbol 0]=0$.
    In addition, the $t$-th value of $\boldsymbol Q_1^E$ is denoted as $q_t$, i.e., $q_t$ is the number of type $t$ clients outside of $B$ that are served by the facility inside $B$. Thus, in order to satisfy the fair constraints, the number of each type clients should satisfy the following constraint:
    for each $t\in [l]$, $\alpha_t\leq \frac{q_t}{\sum_{t=1}^{l} q_t} \leq \beta_t$.
\item If there is only one facility in the block but this facility is not opened, then the table value of this block is $DP[B,0,\boldsymbol 0,\boldsymbol 0,\boldsymbol 0,\boldsymbol 0,\ldots,\boldsymbol 0,\boldsymbol 0]=0$.
\item If there is only one client in the block, then the table value of this block is $DP[B,0,\boldsymbol 0,\boldsymbol Q_1^L,\boldsymbol 0,\boldsymbol 0,\ldots,\boldsymbol 0,\boldsymbol 0]=0$. If the client is a type $t$ ($t\in [l]$) client, then the $t$-th value in $\boldsymbol Q_1^L$ is -1 and the other values in $\boldsymbol Q_1^L$ is 0.
\end{enumerate}

\subsubsection{Computing Table Entries}\label{Sec 3.2.2.}
Now we consider the subproblem on the subtree $T_B$ where block $B$ is at level $i$ of split-tree $T$. For the block $B$, let $B_1,B_2,...,B_u$ be the children of block $B$ where $u$ is at most $2^{O(d)}$. Let $(B,k_B,\boldsymbol Q_1^E,\boldsymbol Q_1^L,\ldots,\boldsymbol Q_m^E,\boldsymbol Q_m^L)$ be the configuration of $B$, and let $(B_j,k_{B_j},\boldsymbol Q_1^{E(j)},\boldsymbol Q_1^{L(j)},\ldots,\boldsymbol Q_m^{E(j)},\boldsymbol Q_m^{L(j)})$ be the configuration of children $B_j$ ($1\leq j\leq u$), where $B_j$ ($j\in [u]$) denotes the $j$-th child of $B$, $k_{B_j}$ is the number of facilities that are opened in $B_j$, $\boldsymbol Q_i^{E(j)}$ ($i\in [m]$) is a vector with $l$ values and each value $q_t^{E(ij)} $ ($0\leq q_t^{E(ij)}\leq n,t\in [l]$) in $\boldsymbol Q_i^{E(j)}$ denotes the number of the $t$-th type clients entering $B_j$ through the $i$-th portal of $B_j$, $\boldsymbol Q_i^{L(j)}$ is a vector with $l$ values and each value $q_t^{L(ij)} $ ($0\leq q_t^{L(ij)}\leq n,t\in [l]$) in $\boldsymbol Q_i^{L(j)}$ denotes the number of the $t$-th type clients leaving $B_j$ through the $i$-th portal of $B_j$.
The values of $k_{B_j}$, $\boldsymbol Q_i^{E(j)}$ and $\boldsymbol Q_i^{L(j)}$ have the following constraints:
\begin{enumerate}
  \item[(1)] $k_{B_1}+k_{B_2}+\ldots+k_{B_u} \leq k_B$;
  \item[(2)] $\sum_{i=1}^{m} {\sum_{j=1}^{u} \boldsymbol Q_i^{E(j)}}+ \sum_{i=1}^{m} {\sum_{j=1}^{u} \boldsymbol Q_i^{L(j)}}$=$\sum_{i=1}^{m} \boldsymbol Q_i^E  + \sum_{i=1}^{m} \boldsymbol Q_i^L$.
\end{enumerate}

\begin{lemma}
Given a subproblem $B$ at level $i$ of split-tree $T$ ($i\neq 0$) and its children $B_1,B_2,\ldots,B_u$ where $u$ is at most $2^{O(d)}$, for the configuration $(B,k_B,\boldsymbol Q_1^E,\boldsymbol Q_1^L,\ldots,\boldsymbol Q_m^E,\boldsymbol Q_m^L)$ of $B$, and the configuration $(B_j,k_{B_j},\boldsymbol Q_1^{E(j)},\boldsymbol Q_1^{L(j)},\ldots,\boldsymbol Q_m^{E(j)},\boldsymbol Q_m^{L(j)})$ of each children $B_j$ ($1\leq j\leq u$), we have
\begin{equation*}
\begin{aligned}
&DP[B,k_B,\boldsymbol Q_1^E,\boldsymbol Q_1^L,\ldots,\boldsymbol Q_m^E,\boldsymbol Q_m^L]
 =\min_{}
\{\sum_{j=1}^{u} DP[B_j,k_{B_j},\boldsymbol Q_1^{E(j)},\boldsymbol Q_1^{L(j)},\ldots,\boldsymbol Q_m^{E(j)},\boldsymbol Q_m^{L(j)}]
+\tau,
\\& with  \ k_{B_1}+k_{B_2}+\ldots+k_{B_u} \leq k_B, \sum_{i=1}^{m} \sum_{j=1}^{u} \boldsymbol Q_i^{E(j)}+ \sum_{i=1}^{m} {\sum_{j=1}^{u} \boldsymbol Q_i^{L(j)}}=\sum_{i=1}^{m} \boldsymbol Q_i^E  + \sum_{i=1}^{m} \boldsymbol Q_i^L \}
\end{aligned}
\end{equation*}
where $\tau$ is 
the cost between portals of $B_1,B_2,\ldots,B_u$ and the ones in $B$, and can be calculated in polynomial time.
\end{lemma}
\begin{proof}
Since there exists consistence relation between the configurations of $B$ and its children, that is,
the number of clients entering or leaving in the configuration $(B,k_B,\boldsymbol Q_1^E,\boldsymbol Q_1^L,\ldots,\boldsymbol Q_m^E,\boldsymbol Q_m^L)$ is equal to the number of clients entering or leaving in the configurations $(B_j,k_{B_j},\boldsymbol Q_1^{E(j)},\boldsymbol Q_1^{L(j)},\ldots,\boldsymbol Q_m^{E(j)},\boldsymbol Q_m^{L(j)})$, we have
\begin{equation*}
\begin{aligned}
\sum_{i=1}^{m} {\sum_{j=1}^{u} \boldsymbol Q_i^{E(j)}}+ \sum_{i=1}^{m} {\sum_{j=1}^{u} \boldsymbol Q_i^{L(j)}}=\sum_{i=1}^{m} \boldsymbol Q_i^E  + \sum_{i=1}^{m} \boldsymbol Q_i^L
\end{aligned}
\end{equation*}

For a particular type $t$ ($t\in [l]$), the number of type $t$ clients entering or leaving in the configuration $(B,k_B,\boldsymbol Q_1^E,\boldsymbol Q_1^L,\boldsymbol Q_2^E,\boldsymbol Q_2^L,\ldots,\boldsymbol Q_m^E,\boldsymbol Q_m^L)$ is equal to the number of type $t$ clients entering or leaving in the configurations $(B_j,k_{B_j},\boldsymbol Q_1^{E(j)},\boldsymbol Q_1^{L(j)},\boldsymbol Q_2^{E(j)},\boldsymbol Q_2^{L(j)},\ldots,\boldsymbol Q_m^{E(j)},\boldsymbol Q_m^{L(j)})$. Thus, we have
\begin{equation*}
\begin{aligned}
\sum_{i=1}^{m} {\sum_{j=1}^{u} q_t^{E(ij)}}+ \sum_{i=1}^{m} {\sum_{j=1}^{u}  q_t^{L(ij)}}=\sum_{i=1}^{m}  q_t^{E(i)}  + \sum_{i=1}^{m} q_t^{L(i)}
\end{aligned}
\end{equation*}

Since $-n\leq q_t^{L(ij)}\leq 0$, $-n\leq q_t^{L(i)}\leq 0$, we have:
\begin{equation}
\label{E35}
\begin{aligned}
\sum_{i=1}^{m} {\sum_{j=1}^{u} q_t^{E(ij)}}- \sum_{i=1}^{m} {\sum_{j=1}^{u}  \mid q_t^{L(ij)}\mid}=\sum_{i=1}^{m}  q_t^{E(i)}  - \sum_{i=1}^{m} \mid q_t^{L(i)}\mid
\end{aligned}
\end{equation}

Then, we have:
\begin{equation}
\label{E40}
\begin{aligned}
\sum_{i=1}^{m} {\sum_{j=1}^{u} q_t^{E(ij)}}+ \sum_{i=1}^{m}  \mid q_t^{L(i)}\mid=\sum_{i=1}^{m} {\sum_{i=1}^{m} q_t^{E(i)}+ \sum_{j=1}^{u}  \mid q_t^{L(ij)}\mid}
\end{aligned}
\end{equation}

Based on the constraints of $k_{B_j}$, $\boldsymbol Q_i^{E(j)}$, $\boldsymbol Q_i^{L(j)}$, to find the minimum cost of the solutions in subproblem $(B,k_B,\boldsymbol Q_1^E,\boldsymbol Q_1^L,\ldots,\boldsymbol Q_m^E, \boldsymbol Q_m^L)$ with table entry $DP[B,k_B,\boldsymbol Q_1^E,\boldsymbol Q_1^L,\ldots,\boldsymbol Q_m^E,\boldsymbol Q_m^L]$, we need to enumerate all consistent subproblems from its children $B_1,B_2,\ldots,B_u$. Moreover, the cost $DP[B,k_B,\boldsymbol Q_1^E,\boldsymbol Q_1^L,\ldots,\boldsymbol Q_m^E,\boldsymbol Q_m^L]$ is also closely related to the cost between portals of $B_1,B_2,\ldots,B_u$ and the ones of $B$, denoted by $\tau$.

We now give the general idea to calculate $\tau$. For each type $t\in [l]$, a weighted bipartite graph $\Phi_t$ is constructed. The value of $\tau$ is obtained based on perfect matching on each bipartite graph constructed.
If we find the minimum cost perfect matching in $\Phi_t$, it gives the optimal case for the type $t$ clients entering or leaving $B_1,\ldots,B_u$ and $B$ through portals of $B_1,\ldots,B_u$ and the ones in $B$.

For each type $t\in [l]$, a weighted bipartite graph $\Phi_t=(R\cup S, E)$ can be constructed by the following way. Let $R=X_E^F\cup X_L^S$, and $S=X_E^S\cup X_L^F$.
For the block $B$, and for each portal $u$ in $B$, if there exist $q$ clients entering $B$ through $u$, then add $q$ new vertices to $X_E^F$. It is easy to see that $\mid X_E^F\mid=\sum_{i=1}^{m} q_t^{E(i)}$.
Similarly, for each portal $u$ in $B$, if there exist $q$ clients leaving $B$ through $u$, then add $q$ new vertices to $X_L^F$. Then, we have $\mid X_L^F\mid=\sum_{i=1}^{m} \mid q_t^{L(i)}\mid$.
For each child $B_j$ of $B$, and for each portal $u$ of $B_j$, if there exist $q$ clients entering $B_j$ through $u$, then add $q$ new vertices to $X_E^S$. Then, we have $\mid X_E^S\mid=\sum_{i=1}^{m} {\sum_{j=1}^{u} q_t^{E(ij)}}$.
Similarly, for each child $B_j$ of $B$, and for each portal $u$ of $B_j$, if there exist $q$ clients leaving $B_j$ through $u$, then add $q$ new vertices to $X_L^S$. Then, we have $\mid X_L^S\mid=\sum_{i=1}^{m} {\sum_{j=1}^{u} \mid q_t^{L(ij)}\mid}$.

For each vertex $v$ in $R\cup S$, let $P(v)$ be the portal in $B$ or $B_j$ such that vertex $v$ is added by the clients entering or leaving block through the portal.

The set $E$ of edges in $\Phi_t$ can be constructed by the following methods.
\begin{enumerate}
\item Add all the edges in $\{(u, v)\mid u \in X_E^F, v\in X_E^S\}$ to $E$.
\item Add all the edges in $\{(u,v)\mid u\in X_L^S,v\in X_L^F\}$ to $E$.
\item Add all edges in $\{(u,v)\mid u\in X_L^S,v\in X_E^S$, where $P(u)\in B_i$ and $P(v)\in B_j$ and $i\neq j\}$ to $E$.
\item For each edge $(u, v)$ in $E$, the weight of edge $(u, v)$ is denoted as $dist(P(u), P(v))$.
\end{enumerate}




\begin{claim}
Given a bipartite graph $\Phi_t=(R\cup S,E)$ as described above, there must exist a perfect matching in $\Phi_t=(R\cup S,E)$.
\end{claim}
\begin{proof}
Assume that there does not exist a perfect matching in $\Phi_t=(R\cup S,E)$. Let $M$ be one of the maximum matchings in $\Phi_t$. Since $M$ is not a perfect matching, there exists at least two unmatched vertices in $\Phi_t$. By equation \ref{E35}, we have $\mid R\mid=\mid S\mid$. Thus, the number of unmatched vertices must be even. Assume that $a, b$ are two unmatched vertices in $\Phi_t$, we have the following cases for the unmatched vertices $a$ and $b$:
Case 1,  $a\in X_E^F, b\in X_E^S$ or $a\in X_L^S, b\in X_L^F$. According to the construction process of $E$, there must exist an edge between $a$ and $b$. Thus, based on $M$, a matching $M'$  with one more edge can be obtained, contradicting the fact that $M$ is a maximum matching in $\Phi_t$.

Case 2, $a\in X_E^F, b\in X_L^F$. We have the following two subcases: 1) There is no edge between vertices in $X_L^S$ and vertices in $X_E^S$. Since $\mid X_E^F\mid=\mid X_E^S\mid, \mid X_L^S\mid=\mid X_L^F\mid$, there must exist two unmatched vertices $a'\in X_E^S$ and $b'\in X_L^S$ such that an augmenting path from $a$ to $a'$ and an augmenting path from $b$ to $b'$ can be found. Thus, based on $M$, a matching $M'$ with larger size can be obtained, contradicting the fact that $M$ is a maximum matching in $\Phi_t$.
2) There exist two vertices $ c\in X_L^S, d\in X_E^S$ such that edge $(c,d)$ is in $E$. According to the construction process of $E$, edges $(a,d)$, $(c,b)$ are contained in $E$. Thus, an augmenting path from $a$ to $b$ can be found, and a matching $M'$ with larger size can be obtained, which is a contradiction.

Case 3, $a\in X_E^S,b\in X_L^S$. We consider the following four cases. 1) There is no edge from $a$ to the vertices in $X_L^S$, and there is no edge from $b$ to the vertices in $X_E^S$. Since $\mid X_E^F\mid=\mid X_E^S\mid,\mid X_L^S\mid=\mid X_L^F\mid$, there must exist two unmatched vertices $a'\in X_E^F$ and $b'\in X_L^F$ such that an augmenting path from $a$ to $a'$ and an augmenting path from $b$ to $b'$ can be found. Thus, based on $M$, a matching $M'$ with larger size can be obtained, contradicting tha fact that $M$ is a maximum matching in $\Phi_t$.
2) There exists a vertex $c$ in $X_L^S$ such that edge $(a,c)$ is in $E$. According to the construction process of $E$, there exist a vertex $d$ in $X_L^F$ such that edges $(c, d)$, $(d, b)$ are in $E$. Thus, an augmenting path from $a$ to $b$ can be found, and a matching $M'$ with larger size can be obtained, which is a contradiction. 3) There exists a vertex $c$ in $X_E^S$ such that edge $(b,c)$ is in $E$. According to the construction process of $E$, there exist a vertex $d$ in $X_E^F$ such that edges $(c, d)$, $(d, a)$ are in $E$. Thus, an augmenting path from $a$ to $b$ can be found, and a matching $M'$ with larger size can be obtained, which is a contradiction. 4) There exists a vertex $c$ in $X_L^S$ such that edge $(a, c)$ is in $E$ and there exists a vertex $d$ in $X_E^S$ such that edge $(b, d)$ is in $E$. Based on $M$, it is easy to find an augmenting path from $a$ to $b$ to get a matching with larger size, which is a contradiction.

Therefore, there must exist a perfect matching in bipartite graph $\Phi_t=(R\cup S,E)$.
\end{proof}

Let $\tau_t$ be the sum weights of the minimum weighted perfect matching in $\Phi_t=(R\cup S,E)$.
\begin{claim}
The value of $\tau$ is $\sum_{t=1}^{l} \tau_t$.
\end{claim}
\begin{proof}
By equation \ref{E35}, we have
\begin{equation}
\label{Eadd}
\begin{aligned}
\sum_{i=1}^{m} {\sum_{j=1}^{u} q_t^{E(ij)}}- \sum_{i=1}^{m}  q_t^{E(i)} =\sum_{i=1}^{m} {\sum_{j=1}^{u}  \mid q_t^{L(ij)}\mid}  - \sum_{i=1}^{m} \mid q_t^{L(i)}\mid
\end{aligned}
\end{equation}
We denote $\sum_{i=1}^{m} {\sum_{j=1}^{u} q_t^{E(ij)}}- \sum_{i=1}^{m}  q_t^{E(i)}$ and $\sum_{i=1}^{m} {\sum_{j=1}^{u}  \mid q_t^{L(ij)}\mid}  - \sum_{i=1}^{m} \mid q_t^{L(i)}\mid$ as the remaining number of clients entering or leaving $B_1,\ldots,B_u$, respectively. Then, the remaining number of clients entering $B_1,\ldots,B_u$ is equal to the remaining number of clients leaving $B_1,\ldots,B_u$.

For each type $t\in [l]$, a weighted bipartite graph $\Phi_t$ is constructed. Let $H_i$ be the number of edges in perfect matching of $\Phi_t$, and let $H=\sum_{i=1}^l H_i$. Since the number of clients entering or leaving in the configuration $(B,k_B,\boldsymbol Q_1^E,\boldsymbol Q_1^L,\ldots,\boldsymbol Q_m^E,\boldsymbol Q_m^L)$ is equal to the number of clients entering or leaving in the configurations $(B_j,k_{B_j},\boldsymbol Q_1^{E(j)},\boldsymbol Q_1^{L(j)},\ldots,\boldsymbol Q_m^{E(j)},\boldsymbol Q_m^{L(j)})$, we can get that the number of pair portals to get $\tau$ value is equal to $H$.

In each graph $\Phi_t$, the perfect matching in $\Phi_t$ gives the relation between the portals of $B_1,\ldots, B_u$ and the ones in $B$, where the clients of type $t$ entering or leaving through those portals. The minimum cost of the solutions in subproblem $(B,k_B,\boldsymbol Q_1^E,\boldsymbol Q_1^L,\boldsymbol Q_2^E,\boldsymbol Q_2^L,\ldots,\boldsymbol Q_m^E, \boldsymbol Q_m^L)$ with table entry $DP[B,k_B,\boldsymbol Q_1^E,\boldsymbol Q_1^L,\boldsymbol Q_2^E, \boldsymbol Q_2^L,\ldots,\\ \boldsymbol Q_m^E,\boldsymbol Q_m^L]$ is based on the value of $\tau$, which is the cost between portals of $B_1,B_2,\ldots,B_u$ and the ones in $B$. Thus, we can get that $\tau \leq \sum_{t=1}^{l} \tau_t$.  We now prove that the value of $\tau$ cannot be less than $\sum_{t=1}^{l} \tau_t$. Assume that $\tau <\sum_{t=1}^{l} \tau_t$. The value of $\tau$ is the cost between portals of $B_1,B_2,\ldots,B_u$ and the ones in $B$, and by the fair constraints, there are $l$ types of clients to deal with. Thus, the value of $\tau$ can be divided into $l$ values, i.e., $\tau=\tau_1'+\ldots+\tau_l'$. By assumption, $\tau <\sum_{t=1}^{l} \tau_t$. Thus, there is a type $t$ such that $\tau_t'<\tau_t$. Let $M$ be the perfect matching in $\Phi_t$ to get the value $\tau_t$, and let $S$ be the set of pair portals used to get the value $\tau_t'$. By the above discussion, the number of pair portals in $S$ is equal to the number of edges in $M$. For each pair portal $(p, p')$ in $S$, by the construction process of edge set $E$, there exists a corresponding edge with weight $dist(p, p')$ in graph $\Phi_t$. Moreover, let $M'$ be the set of corresponding edges of the pair portals in $S$. It is easy to get that $M'$ is a perfect matching in $\Phi_t$ with weight smaller than $M$, contradicting that $M$ is a minimum weighted perfect matching in $\Phi_t$. Thus, $\tau =\sum_{t=1}^{l} \tau_t$.
\end{proof}

Since the number of vertices and the number of edges in graph $\Phi_t$ are polynomial, we can find the minimum cost perfect matching of graph $\Phi_t=(R\cup S,E)$ and compute the value of $\tau$ in polynomial time (\cite{kuhn1955hungarian,munkres1957algorithms}).
\end{proof}

\subsubsection{Assigning Clients to Facilities}\label{ASSi}
Given an instance $(X, dist, \mathcal{C}, \mathcal{F}, k,l, \boldsymbol \alpha, \boldsymbol \beta)$ of the fair $k$-median problem, the above dynamic programming procedure obtains a set $F\subseteq\mathcal{F}$ of size at most $k$ which is a set of open facilities for the fair $k$-median problem and the number of clients assigned to each open facility satisfying the fairness constraints. 
However, for each client $c\in\mathcal{C}$, the dynamic programming procedure only ensures that the number of clients assigned to each open facility satisfies fair constraints, and does not obtain the assignment which maps each client $c\in\mathcal{C}$ to a facility $f\in F$. 
We apply the min-cost max-flow technique to obtain the assignment which maps each client $c\in\mathcal{C}$ to a facility $f\in F$ for fair $k$-median problem. Consider the solution obtained by the above dynamic programming, we obtain a set $F\subseteq\mathcal{F}$ of size at most $k$ which is a set of open facilities for fair $k$-median problem and each $f\in F$ is given some numbers to denote the number of each type clients assigned to it.

Given an open facilities set $F=\{f_1,f_2,\ldots,f_\eta\}$ (where $\eta\leq k$, $F\subseteq\mathcal{F}$) and a vector $\boldsymbol \lambda_i$ with $l$ values for each open facility in $F$, where $\boldsymbol \lambda_i$ ($i\in [\eta]$) is a vector within $l$ values, and each value $\lambda_i^h$ ($0\leq\lambda_h^i\leq n, h\in [l]$) in $\boldsymbol \lambda_i$ denotes the number of the $h$-th type clients that are assigned to $f_i$,  we construct a bipartite graph $G=(U\cup V,A)$ with both capacities and weights on the edges and $b$ values on the vertices as follows.
Let $U=\{c_j^h\mid h\in [l],j\in [\mid \mathcal{C}_h\mid]\}\cup \{s\}$, $V=\{f_i^h\mid i\in [\eta],h\in [l]\}\cup \{t\}$, where $s$ denote the source point, $t$ denotes the sink point, $c_j^h$ denotes the $j$-th client of the type $h$ clients ($c_j^h\in \mathcal{C}_h$), $f_i^h$ denotes the $i$-th facility $f_i\in F$ and we label this facility as type $h$ ($h\in [l]$).

The edges in $G$ can be constructed as follows. For each fixed $h\in [l]$, and for each vertex $c_j^h\in U$ and $f_i^h\in V$ (where $j\in [\mid\mathcal{C}_h\mid],i\in [\eta]$), add a directed edge $e$ from $c_j^h$ to $f_i^h$ in $G$. The capacity of the edge is 1 and the cost of the edge is $dist(c_j^h,f_i)$. Note that parallel edges may exit (parallel edges are kept in $G$). Furthermore, for each vertex $c_j^h\in U-\{s\}$, add a directed edge from $s$ to $c_j^h$ with capacity 0 and cost 0. For each vertex $f_i^h\in V$, add a directed edge from $f_i^h$ to $t$  with capacity 0 and cost 0.

Finally, we set $b$ $value$ for each vertex in bipartite graph $G=(U\cup V,A)$. For each vertex $v\in G$, let $b(v)$ denote the value of each vertex. Then, positive $b(v)$ means that the vertex $v$ supplies $b(v)$ units of flow, and negative $b(v)$ means that the vertex $v$ demands $b(v)$ units of flow. For each vertex $c_j^h\in U$, set $b(c_j^h)=1$. For each vertex $f_i^h\in V$,  set $b(f_i^h)=-\lambda_i^h$. Since $\sum\limits_{i=1}^{\eta} {\sum\limits_{h=1}^{l} b(f_i^h)}=- \mid\mathcal{C}\mid$ and $\sum\limits_{h=1}^{l}{\sum\limits_{j=1}^{\mid\mathcal{C}_h\mid} b(c_j^h)}=\mid\mathcal{C}\mid$
, we have $\sum\limits_{v\in G} {b(v)}=0$. Then, we run a min-cost max-flow algorithm on the bipartite graph $G=(U\cup V,A)$. For the min-cost max-flow found in graph $G$, and for each edge $(c_j^h,f_i^h)$ in $G$, if edge $(c_j^h,f_i^h)$ has a flow of 1, then we assign the client $c_j^h$ to facility $f_i$. Thus, by applying the min-cost max-flow method in graph $G$, we obtain the assignment which maps each client $c\in\mathcal{C}$ to a facility $f\in F$ in polynomial time.

\subsubsection{Analysis and Running Time}
We show that our algorithm outputs a solution of cost at most $(1+\epsilon)$ times the cost of the optimal solution and it gives a QPTAS for the fair $k$-median problem in doubling metrics. Let $OPT$ denote one of the optimal solution of the  fair $k$-median problem in doubling metrics, and let $cost(OPT)$ denote the cost of this optimal solution.

\begin{theorem}
Given an instance $(X, dist, \mathcal{C}, \mathcal{F}, k,l, \boldsymbol \alpha, \boldsymbol \beta)$ of the fair $k$-median problem in doubling metrics, there exists an algorithm such that a solution of cost at most $(1+\epsilon)cost(OPT)$ can be obtained, and the running time is at most $\tilde{O}(n^{{(\frac{\log n}{\epsilon})}^{O(d)}l})$.
\end{theorem}
\begin{proof}
Let $\mathcal{G}$ be the solution obtained by the dynamic programming and clients assignment procedure. For each client $c\in \mathcal{C}$, we say that $f\in\mathcal{G}$ serves a client $c$ if $\mathcal{G}(c)=f$. By lemma \ref{splitpro}, and for any pair of points $u$ and $v$ in the metrics, the probability that $u$ and $v$ are divided into different blocks at level $i$ is at most $O(d)\cdot\frac{dist(u,v)}{2^i}$.
By lemma \ref{5}, for any two points $u$ and $v$ at level $i$, the existence of portal set in each block of split-tree incurs an additional cost of $O(\rho2^{i+1})$, and this is incurred with probability at most $O(d)\cdot\frac{dist(u,v)}{2^i}$. Then, the total additional expected cost for the path between $u$ and $v$ is at most
\begin{equation}
\label{E5}
\begin{aligned}
\sum\limits_{i=0}^{\ell} O(d)(\frac{dist(u,v)}{2^i})\cdot O(\rho2^{i+1})
= \sum\limits_{i=0}^{\ell}O(\rho d\cdot dist(u,v))
= (\ell+1)(O(d\rho dist(u,v))).
\end{aligned}
\end{equation}
Let $\rho=O(\epsilon/d\log n)$, the total additional expected cost for the path between $u$ and $v$ is at most $\epsilon dist(u,v)$.
Consider a pair of points $c$ and $OPT(c)$ ($OPT(c)$ denotes the facility that serves $c$), where $c\in\mathcal C$ is assigned to $OPT(c)$ in the optimal solution $OPT$.
By equation \ref{E5}, 
we have that total additional expected cost for the path between $c$ and $OPT(c)$ is at most $\epsilon dist(c,OPT(c))$. Namely, the total additional expected cost of solution $OPT$ is $\epsilon cost(OPT)$. Let $cost_T(OPT)$ be the total cost of solution $OPT$ based on split-tree $T$. Then, we have that $cost_T(OPT)\leq (1+\epsilon)cost(OPT)$.
Then, we have that $cost_T(\mathcal{G})\leq cost_T(OPT)\leq (1+\epsilon)cost(OPT)$. Thus, the total cost of the solution obtained is no more than $(1+\epsilon)cost(OPT)$.

The running time of our algorithm contains the following three parts: constructing the split-tree $T$, executing the dynamic programming procedure, and executing the min-cost max-flow algorithm. Firstly, the split-tree decomposition within portal set can be found in time $(1/\rho)^{O(d)}n\log \Delta$ (where $\Delta=O(n^4/\epsilon)$).

We now analyze the running time of the dynamic programming process. The running time of the dynamic programming process is related to the number of entries in the dynamic programming table, and the time to calculate the value of each table entry. For a table entry $DP[B,k_B,\boldsymbol Q_1^E,\boldsymbol Q_1^L,\ldots,\boldsymbol Q_m^E,\boldsymbol Q_m^L]$ in the dynamic programming table, we analyze the values of the parameters in table entry. For the given split-tree $T$, the total number of levels in the tree is $O(\log n)$, and there are $n$ leaves in $T$. The number of distinct value of $B$ is equal to the number of nodes in the split-tree $T$, which is at most $O(n\cdot\log n)$.
Moreover, the value of $k_B$  is a non-negative integer bounded by $n$, and the parameters of $q_t^{E(i)}$ and $q_t^{L(i)}$ ($t\in [l]$) in $\boldsymbol Q_i^E$ and $\boldsymbol Q_i^L$ ($i\in [m]$), which are non-negative integers bounded by $n$. Thus, the total number of entries in our dynamic programming table is at most $O(n^{2ml+2}\log n)$.

We now analyze the time to calculate the value of each table entry. For the table entry $DP[B,k_B,\boldsymbol Q_1^E,\boldsymbol Q_1^L,\\ \ldots,\boldsymbol Q_m^E,\boldsymbol Q_m^L]$, we consider all possibilities based on the values of variables $k_{B_j}, \boldsymbol Q_i^{E(j)},\boldsymbol Q_i^{L(j)}$ ($i\in [m],j\in [u]$).
Since each of these $2mlu+u$ variables are non-negative integers bounded by $n$, there are $O(n^{2mlu+u})$ cases to try all possible values of those parameters to satisfy certain constraints. In the process of dynamic programming procedure, the value of $\tau$ is calculated, and it is closely related to the minimum weighted perfect matching in $l$ bipartite graphs, which can be solved in polynomial time (\cite{kuhn1955hungarian,munkres1957algorithms}). Specifically, for any bipartite graph $\Phi_t=(R\cup S, E)$, the number of vertices in $R$ (where $\mid R\mid=\mid S\mid$) is at most $nm\cdot(u+1)$. Thus, the time to calculate the minimum weighted perfect matching of $\Phi_t$ is bounded by $O({(nm\cdot(u+1))}^3)$, and the time to calculate the value of $\tau$ is $O({(nm\cdot(u+1))}^3l)$.

Therefore, the time to calculate the value of each table entry in the dynamic programming process is at most $O(n^{2mlu+u}\cdot(O(1)+O({(nm\cdot(u+1))}^3l)))$ (which can be simplified to $O(n^{2mlu+u+3}m^{3} u^{3}l)$. As the total number of entries in the table is bounded by $O(n^{2ml+2}\log n)$, and the number of portals in a portal set $\rho2^{i+1}$-net of a block at level $i$ is $m=O(d\frac {\log n}{\epsilon})^d$,
the running time of the dynamic programming process is bounded by  $\tilde{O}(n^{{(d\frac{\log n}{\epsilon})}^{O(d)}l})$. Moreover, for the set $F$ of opened facilities obtained by dynamic programming process, the clients can be assigned to $F$ by applying the min-cost max-flow method in $poly(n)$ time.

For the instance $(X, dist, \mathcal{C}, \mathcal{F}, k,l, \boldsymbol \alpha, \boldsymbol \beta)$ of the fair $k$-median problem in doubling metrics, a solution of cost at most $(1+\epsilon)cost(OPT)$  can be obtained in time $\tilde{O}(n^{{(\frac{\log n}{\epsilon})}^{O(d)}l})$.
\end{proof}

\section{An $O(\log k)$-Approximation for Fair $k$-median Problem General Metrics}
In this section, to solve the fair $k$-median problem in general metrics, we first deal with the points in $X$ by a general preprocessing step, then construct a \emph {Hierarchically Separated Trees} (HST) to embed metric space into tree metrics, finally apply a dynamic programming procedure to get the solution.

\subsection{Hierarchically Separated Trees}
Given an instance $(X,dist,\mathcal{C},\mathcal{F},k,l,\alpha,\beta)$ of the fair $k$-median problem, in order to reduce the size of input, we deal with the points in $X$ by the following preprocessing steps and we use $OPT$ to denote an optimal solution of fair $k$-median problem. Find an $O(1)$-approximation feasible solution $L$ of the fair $k$-median problem (\cite{Bercea0KKRS019}), and let $\mathcal{C}(L)$ be the $k$ centers of the feasible solution obtained. Then, for each point $x\in X$, remove $x$, copy a point $x'$ of $x$ and add $x'$ at $\mathcal{C}_x$ where $\mathcal{C}_x$ is the nearest center in $\mathcal{C}(L)$ to $x$. In the new instance obtained, we reduce the size of input locations from $n$ to $k$ with constant factor distortion. Thus, the cost of any solution for this new instance compared with the original instance is increased by at most $O(1)cost(OPT)$.

A \emph{decomposition} of the metric space $(X,dist)$ is a partitioning of $X$ into subsets, which are called blocks. A \emph{hierarchically decomposition} of the metric $(X,dist)$ is a sequence of $\delta+1$ decompositions $D_0,D_1,\ldots,D_{\delta}$ such that every block of $D_{i+1}$ is the union of blocks in $D_i$ and $D_\delta=\{X\}$ and $D_0=\{ {x}\mid x\in X\}$.
The blocks of $D_i$ are called the \emph{level} $i$ blocks. A \emph{hierarchically separated tree} (HST) of the metric space is a complete hierarchically decomposition: the root node is the level $\delta$ block $D_\delta=\{X\}$ and the leaves are the singletons such that $D_0=\{ {x}\mid x\in X\}$. Given a metric space $(X,dist)$, let $\Delta=2^\delta$ be the diameter of the metrics. Starting from $D_{\delta}=\{X\}$, we now discuss the relation between the blocks in $i+1$-th level and $i$-th level, where $0\leq i\leq \delta-1$. Let $r_i=2^i \varrho$ where $\frac{1}{2}\leq \varrho<1$, and let $\pi$ be a random ordering of the points in $X$. For each point $x\in X$, let $\pi(x)$ be the order of $x$ in $\pi$. For each block $S$ at $D_{i+1}$, let $\{y_1, \ldots, y_h\}$ be the set of points in $S$ with order from left to right in $\pi$.
For the point $y_1$, a new block $S^{y_1}$ can be obtained at level $i$ such that $S^{y_1}=\{x\in X\mid dist(x,y_1)\leq r_i\}$.
For any point $y_j$ $(1< j\leq h)$, a new block $S^{y_j}$ can be obtained at level $i$ by the following way: let $Z=\{x\in X\mid dist(x,y_j)\leq r_i\}$.
For each point $x$ in $Z$, if $x$ is not contained in any blocks $\{S^{y_1}, \ldots, S^{y_{j-1}}\}$, then $x$ is contained in $S^{y_j}$.

\begin{lemma}[\cite{fakcharoenphol2004tight}]
\label{HSTpro}
Given a HST of metric space $(X,dist)$ with a sequence of decompositions $D_0,D_1,\ldots,D_{\delta}$, the HST decomposition has the following properties:
\begin{enumerate}
  \item The number of level $\delta$ is $O(\log n)$ (since the aspect ratio of the input metrics $X$ is $O(n^4/\epsilon)$).
  \item Each block of level $i$ has diameter at most $2^{i+1}$, namely the maximum distance between any pair of points in a block of level $i$ is at most $2^{i+1}$.
  \item For each block $S$ of level $i$, there are some edges between the block and its children at level $i-1$, the length from a block $S$ to each of its children at level $i-1$ is equal to $2^i$ (which is also the upper bound of the radius of $S$).
  \item For any pair of points $u,v$ must be separated at level $i$, where $i=(\lfloor\log dist(u,v)\rfloor-1)$.
\end{enumerate}
\end{lemma}

A general metrics can be probabilistically embedded into HST with a distortion of $O(\log \mid\mathcal{X}\mid)$ where $\mid\mathcal{X}\mid$ is the size of input locations and equals to $k$ from the above preprocessing step and this embedding can be constructed in polynomial time (\cite{fakcharoenphol2004tight}). We now define the distance function in HST. For any two blocks $S_i, S_j$ in HST $T$, and for any two points $u\in S_i$, $v\in S_j$, the distance between $u$ and $v$ in $T$ is denoted as the length of the shortest path distance in $T$ between block $S_i$ and block $S_j$.

\subsection{Algorithm for Fair $k$-median Problem in General Metrics}

In this section, we give the algorithm to solve the fair $k$-median problem in general metrics.
\subsubsection{Dynamic Programming States}
Given a HST $T$, the dynamic programming proceeds on $T$ from the leaves to the root. For any block $S$ of HST $T$, let $T_B$ be the subtree rooted at block $B$, and let $\ell(S)$ be the number of children of block $S$, and assume that the set of children of block $S$ is $\{S_1, \ldots, S_{\ell(S)}\}$. For $1\leq i\leq \ell(S)$, let $T_{S,i}$ denote the subtree induced by the blocks in $\{S\}\cup T_{S_1}\cup T_{S_2}\cup\ldots\cup T_{S_i}$, and each subtree $T_{S,i}$ is a subproblem in our dynamic programming process, which includes the information that how many clients of each type entering or leaving $T_{S,i}$.


A table entry in the dynamic programming process is a tuple $DP[S,k_S,i,\boldsymbol Q_S^{enter}, \boldsymbol Q_{Sout}^{enter}, \boldsymbol Q_S^{leave},\boldsymbol Q_{Sout}^{leave}]$, where the parameters are defined as follows:
\begin{enumerate}
  \item $S$ is a block of HST $T$.
  \item $k_S$ is the number of facilities that are opened in the subtree $T_{S,i}$.
  \item $i$ is an integer indicating how many blocks in children are used to get subtree $T_{S,i}$.
  \item $\boldsymbol Q_S^{enter}$ is a vector with $l$ values, and $q_t^E$ ($0\leq q_t^E\leq n$) in $\boldsymbol Q_S^{enter}$ denotes the number of type $t$ clients inside the subtree $T_S\setminus T_{S,i}$ that are served by facilities located in $T_{S,i}$.
  \item $\boldsymbol Q_{Sout}^{enter}$ is a vector with $l$ values, and $q_t^{E(out)}$ ($0\leq q_t^{E(out)}\leq n$) in $\boldsymbol Q_{Sout}^{enter}$ denotes the number of type $t$ clients inside the subtree $T\setminus T_S$ that are served by facilities located in $T_{S,i}$.
  \item $\boldsymbol Q_S^{leave}$ is a vector with $l$ values, and $q_t^L$ ($-n\leq q_t^L\leq 0$) in $\boldsymbol Q_S^{leave}$ denotes the number of type $t$ clients in the subtree $T_{S,i}$ that are served by facilities located inside the subtree $T_S\setminus T_{S,i}$.
  \item $\boldsymbol Q_{Sout}^{leave}$ is a vector with $l$ values, and $q_t^{L(out)}$ ($-n\leq q_t^{L(out)}\leq 0$) in $\boldsymbol Q_{Sout}^{leave}$ denotes the number of type $t$ clients in the subtree $T_{S,i}$ that are served by facilities located inside the subtree $T\setminus T_S$.
\end{enumerate}


The solutions at the root $R$ of HST $T$ are in the table entry $DP[R,k,\ell(R),\boldsymbol 0,\boldsymbol 0,\boldsymbol 0,\boldsymbol 0]$, where $\boldsymbol 0$ is a vector with $l$ zero components. Among all these solutions in $DP[R,k,\ell(R),\boldsymbol 0,\boldsymbol 0,\boldsymbol 0,\boldsymbol 0]$, the algorithm outputs the one with the minimum cost.

The base case of the dynamic programming consist of leaves of the HST where at most one facility can be opened at a given location. Let $T_{S,0}$ denote the simple tree within the singleton $\{S\}$.
Since there is only one facility or client in each leaf of HST, we consider the following three cases for the table entry of leaf node $S$:
\begin{enumerate}
\item For a block $S$ at level 0 of HST $T$, there is only a facility in the block and this facility is opened. The table entry of the subproblem at $S$ is $DP[S,1,0,\boldsymbol 0,\boldsymbol Q_{Sout}^{enter},\boldsymbol 0,\boldsymbol 0]=0$.
    In addition, the $t$-th value in $\boldsymbol Q_{Sout}^{enter}$, denoted as $q_t$, is the number of type $t$ clients inside the subtree $T\setminus T_S$ that are assigned to the facility located in $T_{S}$. Thus, in order to satisfy the fair constraints, for each $t\in [l]$, $\alpha_t\leq \frac{q_t}{\sum_{t=1}^{l} q_t} \leq \beta_t$.
\item For a block $S$ at level 0 of HST $T$, there is only a facility in the block but this facility is not opened. The table entry of the subproblem at $S$ is $DP[S,0,0,\boldsymbol 0,\boldsymbol 0,\boldsymbol 0,\boldsymbol 0]=0$.
\item For a block $S$ at level 0 of HST $T$, there is only a client in the block. The table entry of the subproblem at $S$ is $DP[S,0,0,\boldsymbol 0,\boldsymbol 0,\boldsymbol 0,\boldsymbol Q_{Sout}^{leave}]=0$. Since this client must be a client of type $t$ ($t\in [l]$), the $t$-th value in $\boldsymbol Q_{Sout}^{leave}$ will be -1.
\end{enumerate}

\subsubsection{Computing Table Entries}
Assume that we consider the subproblem on the subtree $T_{S,i}$ where $S$ is a block at level $j$ of HST $T$. Let $(S,k_S,i,\boldsymbol Q_S^{enter},\boldsymbol Q_{Sout}^{enter},\boldsymbol Q_S^{leave},\boldsymbol Q_{Sout}^{leave})$ be the configuration of $T_{S,i}$ with table entry $DP[S,k_S,i,\boldsymbol Q_S^{enter},\boldsymbol Q_{Sout}^{enter},\\ \boldsymbol Q_S^{leave},\boldsymbol Q_{Sout}^{leave}]$, where $1\leq i\leq \ell(v)$. To find the minimum cost of $DP[S,k_S,i,\boldsymbol Q_S^{enter},\boldsymbol Q_{Sout}^{enter},\boldsymbol Q_S^{leave},\boldsymbol Q_{Sout}^{leave}]$, we enumerate all consistent subproblems corresponding to $T_{S_i}$ and $T_{S,i-1}$ (if $i\geq 2$). 
Moreover, the value of the table entries corresponding to $T_{S_i}$ and $T_{S,i-1}$ can be extracted from the dynamic programming table. The recursive expression of the dynamic programming procedure has the following two cases.

Case 1. $i=1$. In this case, we consider the subtree $T_{S,1}$ of $T$ induced by the blocks in $T_{S_1}\cup {S}$. As $S_1$ is the first node in the total order of the children of $S$, the assignments of the subproblem $T_{S,1}$ is the same as the assignments of $T_{S_1}$, and the facilities and clients in any feasible solution to the subproblem $T_{S,1}$ must come from $T_{S_1}$. Thus, we have
\begin{equation*}
\begin{aligned}
&DP[S,k_S,1,\boldsymbol Q_S^{enter},\boldsymbol Q_{Sout}^{enter},\boldsymbol Q_S^{leave},\boldsymbol Q_{Sout}^{leave}] =DP[S_1,k_{S_1},\ell(S_1),\boldsymbol 0,\boldsymbol Q_{Sout}^{enter*},\boldsymbol 0,\boldsymbol Q_{Sout}^{leave*}]+\\&dist_T(S,S_1)\cdot(\sum_{t=1}^l q_t^{E(out)*}+\sum_{t=1}^l \mid q_t^{L(out)*}\mid)
\\& with\ k_{S_1}=k_S, \boldsymbol Q_{Sout}^{enter*}=\boldsymbol Q_S^{enter}+\boldsymbol Q_{Sout}^{enter},\boldsymbol Q_{Sout}^{leave*}=\boldsymbol Q_S^{leave}+\boldsymbol Q_{Sout}^{leave}.
\end{aligned}
\end{equation*}
where $k_{S_1}$ is the number of open facilities in the subtree $T_{S_1}$ satisfying the constraint $k_{S_1}=k_S$. $\boldsymbol Q_{Sout}^{enter*}$ denotes the number of each type clients inside the subtree $T\setminus T_{S_1}$ that are served by facilities inside $T_{S_1}$ and $q_t^{E(out)*}$ in $\boldsymbol Q_{Sout}^{enter*}$ denotes the number of type $t$ ($t\in [l]$) clients inside the subtree $T\setminus T_{S_1}$ that are served by facilities inside $T_{S_1}$, $\boldsymbol Q_{Sout}^{leave*}$ denotes the number of each type clients inside $T_{S_1}$ that are served by facilities inside the subtree $T\setminus T_{S_1}$ and $q_t^{L(out)*}$ in $\boldsymbol Q_{Sout}^{leave*}$ denotes the number of type $t$ ($t\in [l]$) clients inside $T_{S_1}$ that are served by facilities inside the subtree $T\setminus T_{S_1}$. $dist_T(S,S_1)$ denotes the length of the edge $(S,S_1)$.
Moreover, the values of $\boldsymbol Q_{Sout}^{enter*}$ and $\boldsymbol Q_{Sout}^{leave*}$ satisfy the following constraints.
\begin{enumerate}
  \item{(1)}  $\boldsymbol Q_{Sout}^{enter*}$=$\boldsymbol Q_S^{enter}$+$\boldsymbol Q_{Sout}^{enter}$. From the definition of $\boldsymbol Q_{Sout}^{enter*}$ above, the set of clients to get value $\boldsymbol Q_{Sout}^{enter*}$ are either in the subtree $T_S\setminus T_{S_1}$ or in the subtree $T\setminus T_S$. From the table entry, the number of clients in the subtree $T_S\setminus T_{S_1}$ that are served by the facilities in $T_{S,1}$ is in $\boldsymbol Q_S^{enter}$, and the number of clients in the subtree $T\setminus T_S$ that are served by the facilities in $T_{S,1}$ is in $\boldsymbol Q_{Sout}^{enter}$.
  \item{(2)}  $\boldsymbol Q_{Sout}^{leave*}$=$\boldsymbol Q_S^{leave}$+$\boldsymbol Q_{Sout}^{leave}$. From the definition of $\boldsymbol Q_{Sout}^{leave*}$, the set of facilities to get value $\boldsymbol Q_{Sout}^{leave*}$ are either in the subtree $T_S\setminus T_{S_1}$ or in the subtree $T\setminus T_S$. From the table entry, the number of clients in $T_{S_1}$ that are served by the facilities in the subtree $T_S\setminus T_{S_1}$ is in $\boldsymbol Q_S^{leave}$, and the number of clients in $T_{S_1}$ that are served by the facilities in the subtree $T\setminus T_S$ is in $\boldsymbol Q_{Sout}^{leave}$.
\end{enumerate}

Case 2. $2\leq i\leq \ell(S)$. In this case, the value of table entry on the subtree $T_{S,i}$ is obtained from the subproblems on the subtrees $T_{S,i-1}$ and $T_{S_i}$, by considering all various possible values of the parameters $k_S, \boldsymbol Q_S^{enter}, \boldsymbol Q_{Sout}^{enter}, \boldsymbol Q_S^{leave}, \boldsymbol Q_{Sout}^{leave}$. Thus, we have
\begin{equation*}
\begin{aligned}
  &DP[S,k_S,i,\boldsymbol Q_S^{enter},\boldsymbol Q_{Sout}^{enter},\boldsymbol Q_S^{leave},\boldsymbol Q_{Sout}^{leave}]=\min_{
  } \{DP[S,k_S',i-1,\boldsymbol Q_S^{enter'},\boldsymbol Q_{Sout}^{enter'},\boldsymbol Q_S^{leave'},\boldsymbol Q_{Sout}^{leave'}]
   \\&+DP[S_i,k_{S_i},\ell(S_i),\boldsymbol 0,\boldsymbol Q_{S_iout}^{enter*},\boldsymbol 0,\boldsymbol Q_{S_iout}^{leave*}]
  +dist_T(S,S_i)\cdot (\sum_{t=1}^l q_t^{E(out)*}+\sum_{t=1}^l \mid q_t^{L(out)*}\mid)
  \\& \ with\ k_S'+k_{S_i}=k_S,
  \boldsymbol Q_S^{enter'}+\boldsymbol Q_{Sout}^{enter'}+\boldsymbol Q_S^{leave'}+\boldsymbol Q_{Sout}^{leave'}+\boldsymbol Q_{S_iout}^{enter*}+\boldsymbol Q_{S_iout}^{leave*} \\&=\boldsymbol Q_S^{enter}+\boldsymbol Q_{Sout}^{enter}+\boldsymbol Q_S^{leave}+\boldsymbol Q_{Sout}^{leave}
  \}
\end{aligned}
\end{equation*}
where $\boldsymbol Q_S^{enter'}$ denotes the number of each type clients inside the subtree $T_S\setminus T_{S,i-1}$ that are served by facilities inside $T_{S,i-1}$ and $\boldsymbol Q_{Sout}^{enter'}$ denotes the number of each type clients inside the subtree $T\setminus T_S$ that are served by facilities inside $T_{S,i-1}$, $\boldsymbol Q_S^{leave'}$ denotes the number of each type clients inside the subtree $T_{S,i-1}$ that are assigned to facilities inside the subtree $T_S\setminus T_{S,i-1}$ and $\boldsymbol Q_{Sout}^{leave'}$ denotes the number of each type clients inside the subtree $T_{S,i-1}$ that are assigned to facilities inside the subtree $T\setminus T_S$.
Moreover, $\boldsymbol Q_{S_iout}^{enter*}$ denotes the number of each type clients inside the subtree $T\setminus T_{S_i}$ that served by facilities inside $T_{S_i}$, $q_t^{E(out)*}$ in $\boldsymbol Q_{S_iout}^{enter*}$ denotes the number of type $t$ ($t\in [l]$) clients inside the subtree $T\setminus T_{S_i}$ that are served by facilities inside $T_{S_i}$, $\boldsymbol Q_{S_iout}^{leave*}$ denotes the number of each type clients inside $T_{S_i}$ that are served by facilities inside the subtree $T\setminus T_{S_i}$, $q_t^{L(out)*}$ in $\boldsymbol Q_{S_iout}^{leave*}$ denotes the number of type $t$ ($t\in [l]$) clients inside $T_{S_i}$ that are served by facilities inside the subtree $T\setminus T_{S_i}$, and $dist_T(S,S_i)$ denotes the length of the edge $(S,S_i)$.
The values of $\boldsymbol Q_S^{enter'}$, $\boldsymbol Q_{Sout}^{enter'}$, $\boldsymbol Q_S^{leave'}$, $\boldsymbol Q_{Sout}^{leave'}$ and $\boldsymbol Q_{S_iout}^{enter*}$, $\boldsymbol Q_{S_iout}^{leave*}$ should satisfy the following consistent constraints:
\begin{enumerate}
  \item{(1)} $k_S'+k_{S_i}=k_S$. $k_S'$ is the number of open facilities in the subtree $T_{S,i-1}$ and $k_{S_i}$ is the number of open facilities in the subtree $T_{S_i}$ satisfying the constraint $k_S'+k_{S_i}=k_S$.
  \item{(2)} $\boldsymbol Q_S^{enter'}+\boldsymbol Q_{Sout}^{enter'}+\boldsymbol Q_S^{leave'}+\boldsymbol Q_{Sout}^{leave'}+\boldsymbol Q_{S_iout}^{enter*}+\boldsymbol Q_{S_iout}^{leave*}=\boldsymbol Q_S^{enter}+\boldsymbol Q_{Sout}^{enter}+\boldsymbol Q_S^{leave}+\boldsymbol Q_{Sout}^{leave}$.
      The subproblem on the subtree $T_{S,i}$ corresponding to $DP[S,k_S,i,\boldsymbol Q_S^{enter},\boldsymbol Q_{Sout}^{enter},\boldsymbol Q_S^{leave}, \boldsymbol Q_{Sout}^{leave}]$ is consistent with subproblems on the subtrees $T_{S,i-1}$ and $T_{S_i}$ corresponding to $DP[S,k_S',i-1,\boldsymbol Q_S^{enter'},\boldsymbol Q_{Sout}^{enter'},\boldsymbol Q_S^{leave'},\\ \boldsymbol Q_{Sout}^{leave'}]$ and $DP[S_i,k_{S_i},\ell(S_i),\boldsymbol 0,\boldsymbol Q_{S_iout}^{enter*},\boldsymbol 0, \boldsymbol Q_{S_iout}^{leave*}]$.
\end{enumerate}

Finally, based on the solution obtained by the dynamic programming, we can apply min-cost max-flow technique same as the one in doubling metrics to guarantee the assignment function for fair $k$-median problem in general metrics.

\subsubsection{Analysis and Running time}
In this section, we show that our algorithm outputs a solution of cost at most $O(\log k)$ times the cost of the optimal solution in polynomial time. Let $OPT$ denote a optimal solution of the fair $k$-median problem in general metrics, and let $cost(OPT)$ denote the cost of this optimal solution.
\begin{theorem}
Given an instance $(X, dist, \mathcal{C}, \mathcal{F}, k,l, \boldsymbol \alpha, \boldsymbol \beta)$ of the fair $k$-median problem in general metrics, there exists an algorithm such that a solution of cost at most $O(\log k)cost(OPT)$ can be obtained, and the running time is at most $n^{O(l)}$.
\end{theorem}
\begin{proof}
Let $\mathcal G$ be the solution obtained by the dynamic programming process and clients assignment procedure.
For each client $c\in \mathcal{C}$, we say that $f\in \mathcal{G}$ serves the client $c$ if $\mathcal{G}(c)=f$. By lemma \ref{HSTpro}, for any pair of points $u,v$, we have that $u$ and $v$ must be separated at level $i$ where $i=(\lfloor \log dist(u,v)\rfloor-1)$. Then, the total expected cost for the path between $u$ and $v$ on the HST is at most
\begin{equation}
\label{E6}
\begin{aligned}
dist_T(u,v)=2\sum_{j=0}^{i} 2^i=2(2^{i+1}-1)=2\cdot 2^{(\lfloor \log dist(u,v)\rfloor)}-2\leq 2dist(u,v)
\end{aligned}
\end{equation}
Consider a pair of points $c$ and $OPT(c)$ where $c\in\mathcal C$ is assigned to $OPT(c)$ ($OPT(c)$ denotes the facility that serves $c$) in the optimal solution $OPT$. By equation \ref{E6}, we have that the total expected cost for the path between $c$ and $OPT(c)$ is at most $dist_T(c,OPT(c))\leq 2dist(c,OPT(c)$. Namely, the total cost of solution $OPT$ on the tree metrics is at most $cost_T(OPT)\leq 2cost(OPT)$. Then, we have that $cost_T(\mathcal{G})\leq cost_T(OPT)\leq 2cost(OPT)$. Moreover, we have that a metric can be probabilistically embedded into HST with a distortion of $O(\log k)$. Thus, the total cost of the solution $\mathcal{G}$ obtained by the above dynamic programming process is no more than $O(\log k)cost(OPT)$

The running time of our algorithm is the sum of three main parts: constructing the HST $T$, executing the dynamic programming procedure and executing the min-cost max-flow algorithm. Firstly, as described above, we have that a HST can be constructed in polynomial time.

For a table entry $DP[S,k_S,i,\boldsymbol Q_S^{enter},\boldsymbol Q_{Sout}^{enter},\boldsymbol Q_S^{leave},\boldsymbol Q_{Sout}^{leave}]$ in the dynamic programming table, we analyze the the value of the parameters in this table entry.
For the given HST $T$, the total number of levels is $O(\log n)$ and there are $n$ leaves in $T$. The number of distinct value of $S$ is equal to the number of nodes in $T$, which from \cite{fakcharoenphol2004tight} is at most $O(n\cdot\log n)$.
Moreover, the value of parameters $k_S$ is non-negative integer bounded above by $n$. The parameters in $\boldsymbol Q_S^{enter},\boldsymbol Q_{Sout}^{enter},\boldsymbol Q_S^{leave},\boldsymbol Q_{Sout}^{leave}$ are also non-negative integers bounded above by $n$. For a fixed $S$, each $i$ denotes the first $i$ children of $S$ from the total ordering of the children of $S$ in $T$. Since $T$ has exactly $n$ leaves, $S$ can have no more than $n$ children. Thus, $i$ is bounded above by $n$. From the above arguments, the total number of table entries in our dynamic programming table is at most $O(n^{4l+3}\log n)$.

We now analyze the time to calculate the value of each table entry. For a table entry $DP[S,k_S,i,\boldsymbol Q_S^{enter},\\ \boldsymbol Q_{Sout}^{enter},\boldsymbol Q_S^{leave},\boldsymbol Q_{Sout}^{leave}]$, we consider all possibilities based on the parameters $k_S',k_{S_i},\boldsymbol Q_S^{enter'},\boldsymbol Q_{Sout}^{enter'},\boldsymbol Q_S^{leave'},\\ \boldsymbol Q_{Sout}^{leave'},\boldsymbol Q_{S_iout}^{enter*},\boldsymbol Q_{S_iout}^{leave*}$. Since each of these $O({6l+2})$ parameters are non-negative integers bounded above by $n$, there are $O(n^{6l})$ cases to try all possible values of these parameters to satisfy certain constraints.

Thus, the running time of our dynamic programming algorithm is bounded by $O(n^{4l+3}\log n\cdot n^{6l})=O(n^{10l}\log n)$. Moreover, for the set $F$ of opened facilities obtained by dynamic programming process, the clients can be assigned to $F$ by applying the min-cost max-flow method in $poly(n)$ time.

For the instance $(X, dist, \mathcal{C}, \mathcal{F}, k,l, \boldsymbol \alpha, \boldsymbol \beta)$ of the fair $k$-median problem in general metrics, a solution of cost at most $O(\log k)cost(OPT)$  can be obtained in time $n^{O(l)}$.

\end{proof}

\section{Conclusions}
In this paper, we study the fair $k$-median problem in doubling metrics and general metrics. We present a $(1+\epsilon)$-approximation algorithm that runs in QPTAS. We further show that our algorithm works in a more general metrics and we present an $O(\log k)$-approximation algorithm that runs in time $n^{O(l)}$.
Our two approximation algorithms for the fair $k$-median problem are the first results for the corresponding problems without any fair violation, respectively.


\bibliographystyle{abbrvnat}
\bibliography{fair-arx}

\end{document}